\documentclass[twocolumn]{svjour3}

\usepackage{bm}

\usepackage{amsmath,amssymb,bm}
\usepackage{multirow,bigstrut}
\usepackage[retainorgcmds]{IEEEtrantools}
\usepackage{graphicx}
\usepackage{url}

\usepackage{hyperref}

\usepackage{multirow}
\graphicspath{{figures-pdf/}}
\newlength\imagewidth
\usepackage{color}
\definecolor{dgreen}{rgb}{0,.6,0}
\usepackage[normalem]{ulem}
\usepackage{cite}
\usepackage{amsmath}
\usepackage[nolabel]{fnbreak}
\journalname{Nonlinear Dynamics}

\usepackage{lineno}
\emergencystretch=\maxdimen
\begin{document}

\title{Breaking a novel image encryption scheme based on improved hyperchaotic sequences}

\author{Chengqing Li \and Yuansheng Liu \and Tao Xie \and Michael Z. Q. Chen}

\authorrunning{C. Li et al.}

\institute{
Chengqing Li, Yuansheng Liu, Tao Xie\at
              MOE Key Laboratory of Intelligent Computing and Information Processing,
              College of Information Engineering,
              Xiangtan University, Xiangtan 411105, Hunan, China \\
              Tel.: +86-731-52639779\\
              Fax: +86-731-58292217\\
              \email{DrChengqingLi@gmail.com}, \url{http://orcid.org/0000-0002-5385-7644}\\
Michael Z. Q. Chen \at
Department of Mechanical Engineering,\\
The University of Hong Kong, Hong Kong
}

\date{Received: Apr 16, 2013}

\maketitle

\begin{abstract}
Recently, a novel image encryption scheme based on improved hyperchaotic sequences was proposed. A pseudo-random number sequence, generated
by a hyper-chaos system, is used to determine two involved encryption functions, bitwise exclusive or (XOR) operation and modulo addition. It was reported
that the scheme can be broken with some pairs of chosen plain-images and the corresponding cipherimages. This paper re-evaluates the security of the encryption
scheme and finds that the encryption scheme can be broken with only one known plain-image. The performance of the known-plaintext attack, in terms
of success probability and computation load, become even much better when two known plain-images are available. In addition, security defects on insensitivity
of the encryption result with respect to changes of secret key and plain-image are also reported.


\end{abstract}

\section{Introduction}

The popularization of image capture devices and fast improvement of transmission speed over all kinds
of networks makes security of images become more and more important. However, the traditional text encryption
techniques cannot protect images efficiently due to the big difference between images and texts. The subtle similarities between chaos and cryptography,
e.g. sensitivity to initial conditions/control parameter of a chaotic system is very similar to diffusion with a small change in the 	
plaintext/secret key of a cryptography system, attract researchers to consider chaos as a novel way to design secure and efficient encryption schemes \cite{YaobinMao:CSF2004,Ye:Scramble:PRL10,ChenJY:Joint:TCSII11}. Meanwhile,
some cryptanalysis work demonstrated that some chaos-based encryption schemes are insecure against various
conventional attacks to different extents from the viewpoint of modern cryptology \cite{LiShujun:YTSCipher:IEEETCASII2004,Xiao:ImproveTable:TCASII06,
David:AttackingShuffling:CSF09,Li:AttackingIVC2009,SolakErcan:AnaFridrich:BAC10,Wangxy:shuffling:OC11,Zhang:perceptron:ND12,cqli:Breakwxy:ND12}. Some general approaches evaluating security of chaos-based encryption schemes were summarized in \cite{AlvarezLi:Rules:IJBC2006,ShujunLi:ChaosBook2011}.

In \cite{Zhu:hyperchaotic:OC12}, a novel image encryption scheme based on improved hyperchaotic sequences was proposed, where
a pseudo-random number sequence (PRNS), generated by a four-dimensional hyper-chaos system, is used to control the modulation addition and
the bitwise exclusive OR operation. Shortly after the publication of \cite{Zhu:hyperchaotic:OC12}, \textit{Fatih et al.} found that an equivalent secret key of the encryption scheme
can be obtained by a brute-force method when some chosen plain-images and the corresponding cipher-images are available \cite{Fatih:breakhyperchaotic:OC12}. The present paper re-evaluates the
security of the encryption scheme proposed in \cite{Zhu:hyperchaotic:OC12}, and
discovers the following security problems: (1) the scope of the equivalent secret key of the encryption scheme can be narrowed efficiently by
comparing one known plain-image and the corresponding cipher-image; (2) the equivalent secret key can be easily confirmed when two known plain-images and the corresponding cipher-images
are available; (3) encryption results are not sensitive with respect to changes of the plain-images/secret key.

The rest of this paper is organized as follows. Next section briefly introduces the image encryption scheme under study. Section~\ref{sec:cryptanalysis}
reviews the cryptanalysis work proposed by \textit{Fatih et al.} and then present an efficient known-plaintext attack on the image encryption scheme under study
in detail with some experimental results. The last section concludes the paper.

\section{The image encryption scheme under study}

The plaintext of the image encryption scheme under study is a gray scale image. Without loss of generality,
the plain-image can be represented as a one-dimensional
8-bit integer sequence $\bm{P}=\{p(i)\}_{i=1}^L$ by scanning it
in the raster order, where $L$ is the number of pixels of the plain-image, and $L$ is assumed to be a multiple of 4. Correspondingly, the cipher-image is denoted by
$\bm{C}=\{c(i)\}_{i=1}^L$. Then, the proposed image encryption scheme can be described as follows\footnote{For the sake of completeness, some notations in the original paper \cite{Zhu:hyperchaotic:OC12}
are modified under the condition that the encryption scheme is not changed.}:

\begin{itemize}
\item \textit{The secret key}:
initial state $(x(0), y(0), z(0), w(0))$ of the hyperchaotic system proposed in \cite{Niu:hyperchaotic:CNSNS10}, which is given as
\begin{equation}
\label{HyperchaotiSystem}
\begin{cases}
\dot{x}=a(y-x)+yz,\\
\dot{y}=cx-y-xz+w,\\
\dot{z}=xy-bz,\\
\dot{w}=dw-xz,
\end{cases}
\end{equation}
where $(a, b, c, d)=(35, 8/3, 55, 1.3)$.

\item\textit{The initialization procedure}:

(1) In double-precision floating-point arithmetic, solve Eq.~(\ref{HyperchaotiSystem}) with the fourth order Runge-Kutta method with a fixed step
length, $h=0.001$, $N_{0}$ times from the initial condition $(x(0), y(0), z(0), w(0))$ iteratively, where $N_{0}>500$.

(2) Iterate the above quantization process $L/4$ more times and obtain a four-dimensional state sequence
$\{(x(i), y(i), z(i), w(i))\}_{i=1}^{L/4}$.

(3) Generate PRNS $K=\{k(i)\}_{i=1}^L$ as follows: for $l=1, 2, \dots, L/4$,
set $k(4l-3)=F(x(l))$, $k(4l-2)=F(y(l))$, $k(4l-1)=z(x(l))$, and $k(4l)=F(w(l))$,
where
\begin{equation*}
F(x)=\left(\lfloor(\left|G(x)\right|-\lfloor \left|G(x)\right| \rfloor)\times 10^{14}\rfloor\right)\bmod 256,
\label{eq:transform}
\end{equation*}
$G(x)=x\times 10^2-[x\times 10^2]$, and $|x|$, $[x]$ and $\lfloor x \rfloor$ round $x$ to the absolute value of $x$, the nearest integers of $x$ and the nearest integers less than or equal to $x$, respectively. Note that
\begin{equation*}
F(x)=(|x\times 10^{2}-[x\times 10^{2}]|\times 10^{14})\bmod 256
%\label{eq:transform2}
\end{equation*}
since $\lfloor |G(x)| \rfloor)\equiv0$.

\item\textit{The encryption procedure} includes the following two rounds of confusion steps.

(1) \textit{Confusion I:} for $i=2\sim L$, do
\begin{equation}
t(i)=p(i)\oplus k(i-1)\oplus(t(i-1)\dotplus k(i)),
\label{eq:encryptR1}
\end{equation}
where
\begin{equation}
t(1)=p(1)\oplus k(1)\oplus(c(0)\dotplus k(1)),
\label{eq:firstelementR1}
\end{equation}
$c(0)$ is a predefined integer falling within the interval $[1, 255]$.

(2) \textit{Confusion II:} for $i=2 \sim L$, do
\begin{equation}
c(i)=t(i)\oplus k(i-1)\oplus (c(i-1)\dotplus k(i)),
\label{eq:encryptR2}
\end{equation}
where
\begin{equation}
c(1)=t(1)\oplus k(1)\oplus(t(L)\dotplus k(1)).
\label{eq:firstelementR2}
\end{equation}

\item[--] \textit{The decryption procedure} is similar to the encryption procedure except the following points:
(1) \textit{Confusion II} is performed first; (2) the operation on each elements in both of the two confusion steps
is carried out in a reverse order; (3) the variables $t(i)$ and $p(i)$ in Eq.~(\ref{eq:encryptR1}) and the variables $c(i)$ and $t(i)$ in Eq.~(\ref{eq:encryptR2}) are swapped, respectively.
\end{itemize}

\section{Cryptanalysis}
\label{sec:cryptanalysis}

\subsection{Attack proposed by \textit{Fatih et al.}}

To make presentation of this paper more complete, \textit{Fatih et al.}'s attack proposed in \cite{Fatih:breakhyperchaotic:OC12}
is reviewed and commented in this subsection.

Substituting Eq.~(\ref{eq:firstelementR1}) and Eq.~(\ref{eq:encryptR1}) into Eq.~(\ref{eq:firstelementR2}) and Eq.~(\ref{eq:encryptR2}), respectively,
one has
\begin{IEEEeqnarray}{rCl}
c(1)& = & p(1)\oplus k(1)\oplus(c(0)\dotplus k(1))\oplus k(1)\oplus(t(L)\dotplus k(1)) \nonumber\\
    & = & p(1)\oplus(c(0)\dotplus k(1))\oplus(t(L)\dotplus k(1)) \label{eq:condition1}
\end{IEEEeqnarray}
and
\begin{IEEEeqnarray}{rCl}
c(i)& = & p(i)\oplus k(i-1)\oplus(t(i-1)\dotplus k(i))\oplus k(i-1)\nonumber\\
    & &\oplus (c(i-1)\dotplus k(i)) \nonumber\\
    & = & p(i)\oplus(t(i-1)\dotplus k(i))\oplus (c(i-1)\dotplus k(i))\label{eq:condition2}
\end{IEEEeqnarray}
for $i=2\sim L$. The idea of \textit{Fatih et al.}'s attack is to search
$(t(L), k(1))$ and $(t(i-1), k(i))$ and verify them with Eq.~(\ref{eq:condition1}) and Eq.~(\ref{eq:condition2}), respectively,
where $i=2\sim L$. In \cite{Fatih:breakhyperchaotic:OC12}, \textit{Fatih et al.} choose a plain-image of fixed value zero, namely $p(i)\equiv 0$.
Success of \textit{Fatih et al.}'s attack depends on whether the known values of $\alpha$ and $y$ can verify the combination of $\beta$ and $x$ in
\begin{equation}
y=(\alpha\dotplus x)\oplus (\beta \dotplus x),
\label{eq:essentialfunction}
\end{equation}
where $\alpha, \beta, x$ and $y$ are all $8$-bit integers, and $(\alpha\dotplus x)=(\alpha+x)\bmod 2^8$. Referring to \cite{Cqli:breakmodulo:IJBC13}, one can see that it is very difficult to estimate
the required number of known/chosen plain-images assuring the success of \textit{Fatih et al.}'s attack. In addition, the computational complexity of
Fatih \textit{et al.}'s attack is $O(L\cdot 256\cdot 256\cdot 4\cdot 4)=O(2^{20}L)$, which means the attack complexity is high when $L$ is very large.

\subsection{Attack with one known plain-image}

In \cite[Sec.~3.4]{Zhu:hyperchaotic:OC12}, it was claimed that the image encryption scheme under study is robust against known/chosen-plaintext attack. However, we found the encryption scheme can be broken with even only one known plain-image.

\begin{proposition}
Assume that one pair of known plain-image, $\bm{P}=\{p(i)\}_{i=1}^L$, and the corresponding cipher-image, $\bm{C}=\{c(i)\}_{i=1}^L$, are available,
then the unknown sequences $\{t(i)\}_{i=1}^{L}$ and $\{k(i)\}_{i=1}^{L-2}$ are only determined by the values of $k(L-1)$ and $k(L)$.
\label{propo:determinesequence}
\end{proposition}
\begin{proof}
Given the values of $k(L-1)$ and $k(L)$, from Eq.~(\ref{eq:encryptR2}) one can obtain
\begin{equation}
t(L)=c(L)\oplus k(L-1)\oplus(c(L-1)\dotplus k(L)).
\label{eq:tL}
\end{equation}
Incorporating Eq.~(\ref{eq:encryptR1}) into Eq.~(\ref{eq:encryptR2}), one has
\begin{equation}
t(L-1)=(t(L)\oplus p(L)\oplus k(L-1))\dot{-}k(L),
\label{eq:tL_1}
\end{equation}
where $a\dot{-}b=(a-b+256)\bmod 256$.
Then, one can obtain
\begin{equation}
k(L-2)=c(L-1)\oplus t(L-1)\oplus(c(L-2)\dotplus k(L-1)).
\end{equation}
Similarly, one can obtain
\begin{equation}
\begin{cases}
t(i)=(t(i+1)\oplus p(i+1)\oplus k(i))\dot{-}k(i+1)\\
k(i-1)=c(i)\oplus t(i)\oplus(c(i-1)\dotplus k(i))
\end{cases}
\end{equation}
for $i=L-2\sim 2$, and
\begin{equation*}
t(1)=(t(2)\oplus p(2)\oplus k(1))\dot{-}k(2).
\end{equation*}
Therefore, the proposition is proven.
\end{proof}

From Proposition~\ref{propo:determinesequence}, one can see that the equivalent secret key of the image encryption scheme under study,
$\{t(i)\}_{i=1}^{L}$ and $\{k(i)\}_{i=1}^{L}$,
are only determined by the values of $k(L-1)$ and $k(L)$ when one pair of known-plaintext and the corresponding cipher-text are available.
As $t(L)$ is determined by $k(L-1)$ and $k(L)$ via Eq.~(\ref{eq:tL}),
and $t(1)$ and $k(1)$ are generated by them in the above iteration form, two independent equations in the form of Eq.~(\ref{eq:essentialfunction}), Eq.~(\ref{eq:firstelementR1}) and Eq.~(\ref{eq:firstelementR2}), are available for verification of the search in this attack method. Success of this attack depends on whether a
wrong version of $(k(L-1), k(L))$ can generate the corresponding version of $(t(1), t(L), k(1))$ passing the verification of Eq.~(\ref{eq:firstelementR1}) and Eq.~(\ref{eq:firstelementR2}).
Assume that $t(1)$, $t(L)$ and $k(1)$ satisfy an uniform distribution, the probability of passing verification of Eq.~(\ref{eq:firstelementR1}) and Eq.~(\ref{eq:firstelementR2})
are both $\frac{1}{256}$. Therefore, only a small number of $k(L-1)$ and $k(L)$ can pass the verification. As shown in Sec.~\ref{sec:defect}, $\{k(i)\}_{i=1}^{L}$ and $\{k(i)\oplus 128\}_{i=1}^{L}$ are equivalent for encryption/decryption (excluding the most significant bit plane) of the image encryption scheme under study, they are considered as the same one in this section.
Note that $\{c(i)\}_{i=1}^{L}$, $\{k(i)\}_{i=1}^{L}$ and $L$ all have influence on the verification, the success rate is very hard to be estimated.
To illustrate this problem, the image ``Peppers" of size $512\times 512$, shown in Fig.~\ref{figure:plaintextattack1}a),
is chosen as the known plain-image, the number of possible versions of $\{k(i)\}_{i=1}^{L}$ passing the verification under one hundred
random secret keys are shown in Fig.~\ref{figure:numbersolution}. As for $5\%$ of the one hundred random secret keys, the equivalent secret key can be confirmed definitely. As for more than $70\%$ of them, the scope size of equivalent secret key is less than 6. When $(x(0), y(0), z(0), w(0))=(5, 10, 5, 10)$, $N_0=1000$ and $c(0)=3$ (the key used in \cite[Sec.~3]{Zhu:hyperchaotic:OC12}), one of the possible versions of $\{k(i)\}_{i=1}^{L}$ passing the verification
is used to decrypt the cipher-image shown in Fig.~\ref{figure:plaintextattack1}, and the result is shown in Fig.~\ref{figure:plaintextattack1}d).
It is counted that $60.13\%$ of the pixels of the image shown in Fig.~\ref{figure:plaintextattack1}d) are correct, which
shows that even the wrong version may be used to recover some information of the cipher-image. Therefore, we can conclude that
this attack is very effective. From Proposition~\ref{prop:xor128}, one can see that
$(k(L-1), k(L))=(a, b)$ and $(k(L-1), k(L))=(a, b\oplus 128)$ are equivalent for Eq.~(\ref{eq:tL_1}). Therefore, the computation complexity of this attack can be estimated as $O(256\cdot 128\cdot L\cdot 2\cdot 3)=O(2^{17}L)$.
\begin{figure}[!htb]
\centering
\begin{minipage}[t]{2.5\imagewidth}
\includegraphics[width=2.5\imagewidth]{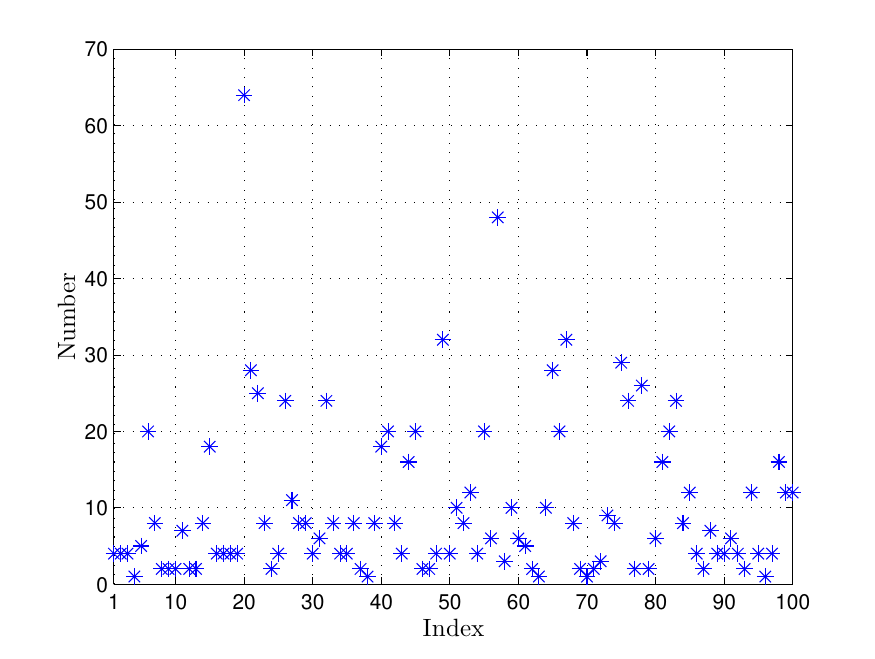}
\end{minipage}
\caption{The number of possible versions of $\{k(i)\}_{i=1}^{L}$ passing the verification under every
set of random secret key.}
\label{figure:numbersolution}
\end{figure}

\begin{figure}[!htb]
\centering
\begin{minipage}[t]{\imagewidth}
\includegraphics[width=\imagewidth]{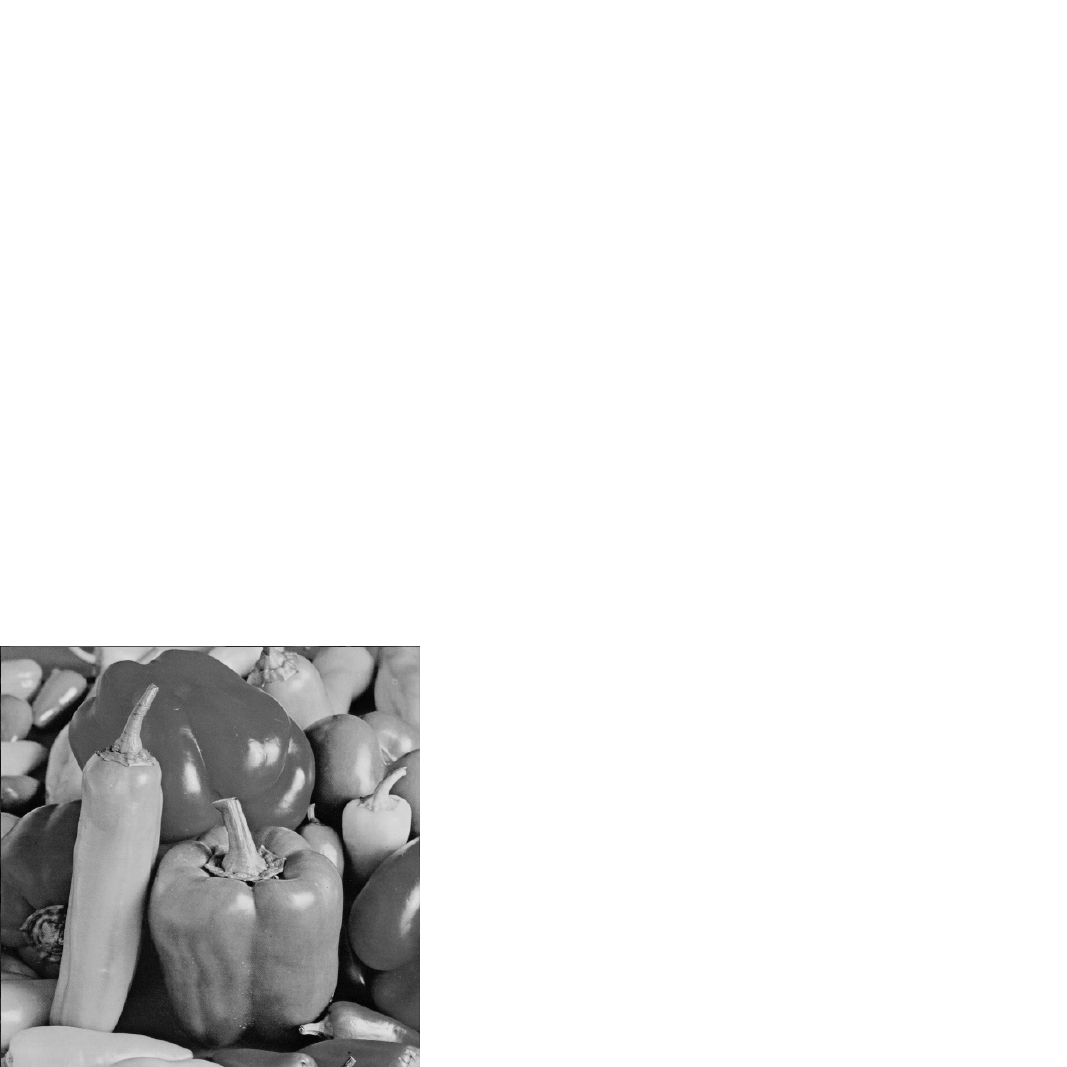}
a)
\end{minipage}
\begin{minipage}[t]{\imagewidth}
\includegraphics[width=\imagewidth]{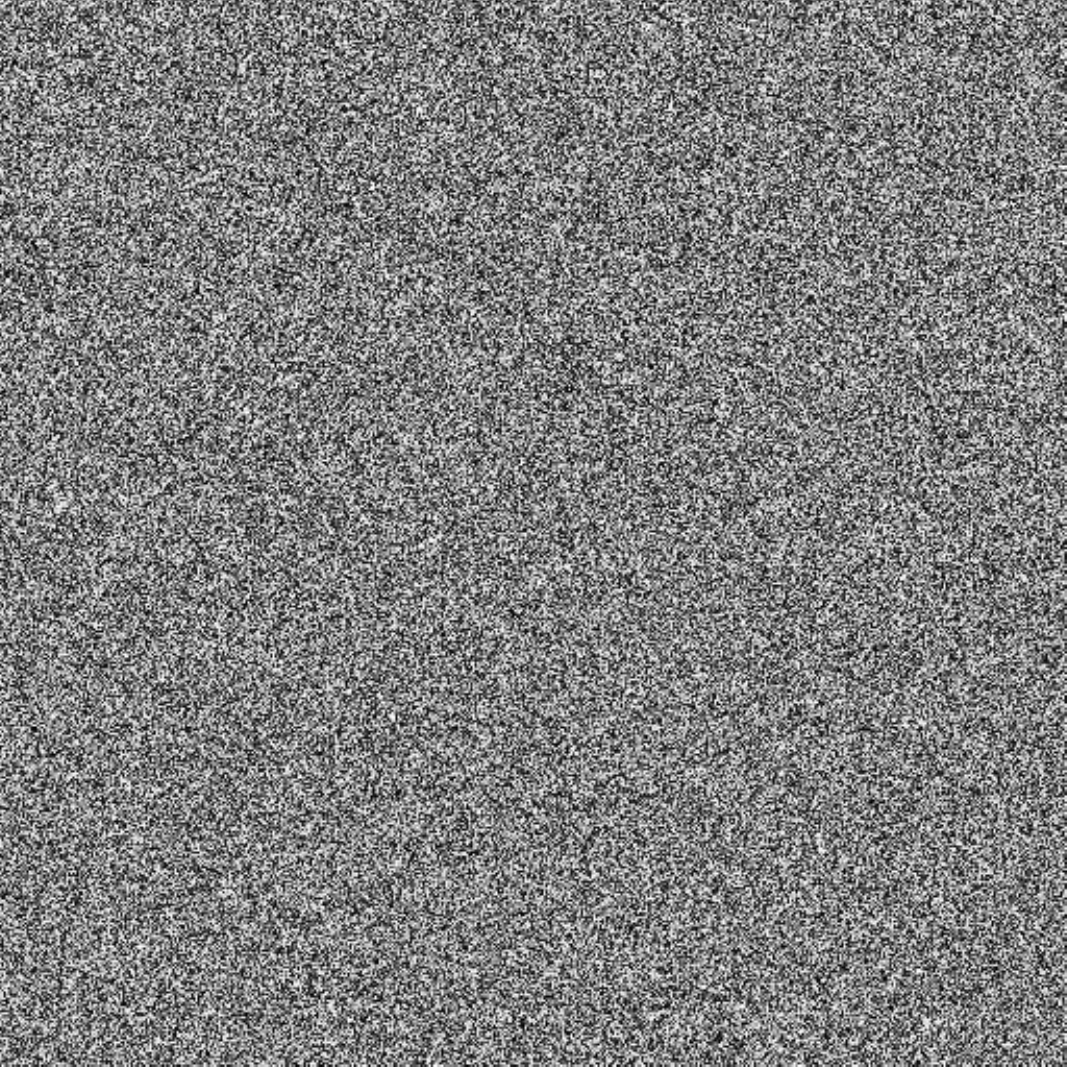}
c)
\end{minipage}
\\
\begin{minipage}[t]{\imagewidth}
\includegraphics[width=\imagewidth]{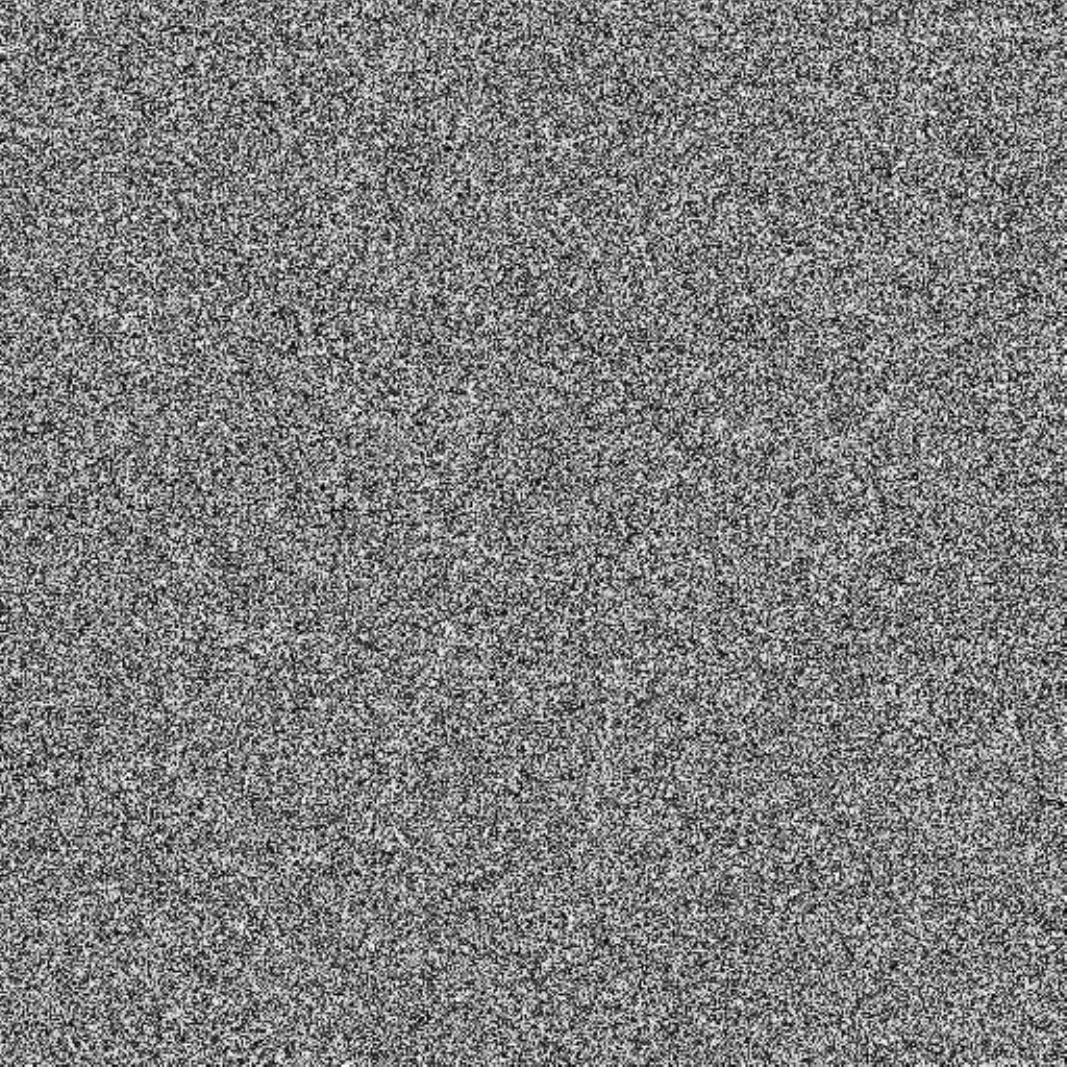}
b)
\end{minipage}
\begin{minipage}[t]{\imagewidth}
\includegraphics[width=\imagewidth]{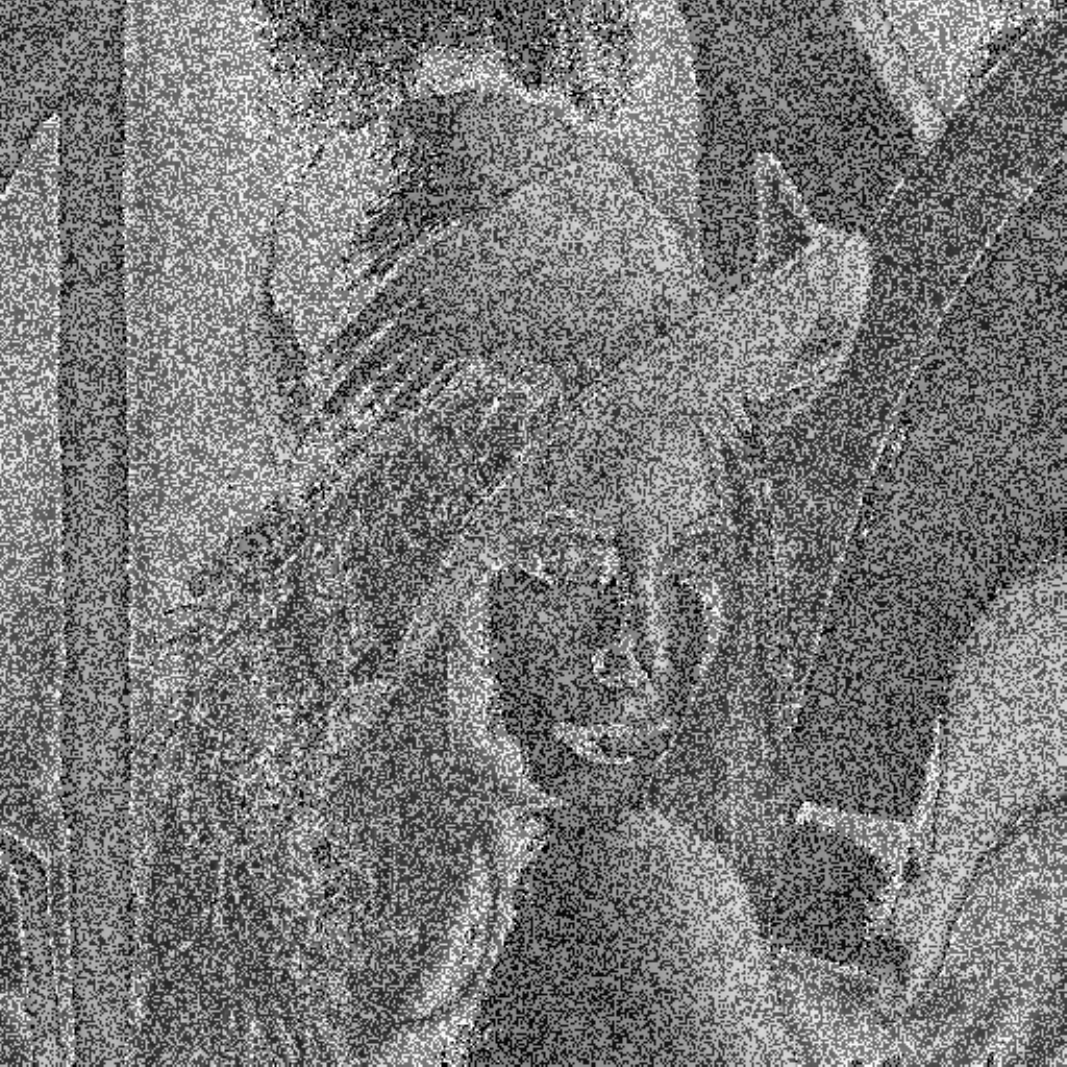}
d)
\end{minipage}
\caption{The known-plaintext attack I: a) known plain-image ``Peppers"; b) the cipher-image of ``Peppers"; c) the cipher-image of a plain-image ``Lenna";
d) decryption result of Fig.~\ref{figure:plaintextattack1}c).}
\label{figure:plaintextattack1}
\end{figure}

\subsection{Attack with two known plain-images}

When two known plain-images, $\bm{P}_1=\{p_1(i)\}_{i=1}^L$ and $\bm{P}_2=\{p_2(i)\}_{i=1}^L$, and the corresponding cipher-images,
$\bm{C}_1=\{c_1(i)\}_{i=1}^L$, $\bm{C}_2=\{c_2(i)\}_{i=1}^L$, are available, coincidence of two versions of $\{k(i)\}_{i=L-2}^{1}$ can be used as
$L-2$ independent conditions to verify the search of $(k(L-1), k(L))$ in the above sub-section. Therefore, the success probability of obtaining the equivalent secret key
can be improved greatly and the attack complexity can be much reduced at the same time.

The detailed approach of the attack can be described as follows.

\begin{itemize}
\item \textit{Step 1)} Set $i=L-1$ and $(k(L-1), k(L))$ with a possible set of values and obtain
\begin{equation*}
t_1(L)=c_1(L)\oplus k(L-1)\oplus(c_1(L-1)\dotplus k(L))
%\label{eq:tL11}
\end{equation*}
and
\begin{equation*}
t_2(L)=c_2(L)\oplus k(L-1)\oplus(c_2(L-1)\dotplus k(L)).
%\label{eq:tL22}
\end{equation*}

\item \textit{Step 2)} Set $i=i-1$. If $i>1$ and
\begin{IEEEeqnarray}{rCl}
c_1(i)\oplus t_1(i)\oplus(c_1(i-1)\dotplus k(i))&= & c_2(i)\oplus t_2(i) \nonumber\\
\oplus(c_2(i-1)\dotplus k(i)),
\label{eq:conditionki}
\end{IEEEeqnarray}
repeat \textit{Step 2)}; otherwise go to \textit{Step 1)},
where
\begin{equation}
\begin{cases}
t_1(i)=(t_1(i+1)\oplus p_1(i+1)\oplus k(i))\dot{-}k(i+1),\\
t_2(i)=(t_2(i+1)\oplus p_2(i+1)\oplus k(i))\dot{-}k(i+1).
\end{cases}
\end{equation}

\item \textit{Step 3)} If $i=1$,
\begin{equation}
c_1(1)=t_1(1)\oplus k(1)\oplus(t_1(L)\dotplus k(1))
\label{eq:condtionkL1}
\end{equation}
or
\begin{equation}
c_2(1)=t_2(1)\oplus k(1)\oplus(t_2(L)\dotplus k(1)),
\label{eq:condtionkL2}
\end{equation}
output the value of $(k(L-1), k(L))$; otherwise go to \textit{Step 1)}.
\end{itemize}

Now, let's analyze the performance of the above attack. Observe Eq.~(\ref{eq:conditionki}), one has
\begin{equation*}
Prob(t)=\prod_{i=L-2}^tProb(i),
\end{equation*}
where $Prob(i)$ denotes the probability of condition~(\ref{eq:conditionki}) being satisfied, and
$t\in\{L-2, L-3, \cdots, 1\}$. Obviously, Eq.~(\ref{eq:conditionki}) can be considered as a function in the form of Eq.~(\ref{eq:essentialfunction}).
Given variable $\alpha, \beta, x, y$ of uniform distribution, the probability that Eq.~(\ref{eq:essentialfunction}) holds is $1/256$. Assume $\{t(i)\}_{i=1}^L$, $\{c(i)\}_{i=1}^L$ and $\{k(i)\}_{i=1}^L$ distribute uniformly, one can get $Prob(t)=(1/256)^{L-1-t}$.
Therefore, one can assure that $\{k(i)\}_{i=1}^{L-1}$, $\{t_1(i)\}_{i=1}^{L}$ and $\{t_2(i)\}_{i=1}^{L}$ can be determined in a very extremely high probability when the variable $i$ in \textit{Step 3)} can reach to $L-5$. Once $t_1(1), t_1(L), t_2(1), t_2(L), k(1)$ are determined, the remaining values of $k(L)$ can be further confirmed with condition~(\ref{eq:condtionkL1}) or condition~(\ref{eq:condtionkL2}). In addition, Eq.~(\ref{eq:firstelementR2}) can also be used for verification. Now, one can conclude that
$\{k(i)\}_{i=1}^{L-1}$ can be determined with an extremely high probability when $L\ge 5$. The computational complexity of this attack is $O(256\cdot 128\cdot 5\cdot (3\cdot 2+7)+L\cdot 2\cdot 3)=O(2^{21}+6L)$,
which is much smaller than that of \textit{Fatih et al.}'s attack.

To verify the above analysis, some experiments were performed. Beside the pair of known plain-image and the corresponding
cipher-image shown in Fig.~\ref{figure:plaintextattack1}, another plain-image ``Babarra" and the encrypted version, shown in Fig.~\ref{figure:plaintextattack}a), and Fig.~\ref{figure:plaintextattack}b), respectively, are used. Then, the obtained equivalent secret key
is used to decrypt the cipher-image shown in Fig.~\ref{figure:plaintextattack1}c) and the recovery result is shown in Fig.~\ref{figure:plaintextattack}c), which is identical with the original version.

\begin{figure}[!htb]
\centering
\begin{minipage}[t]{\imagewidth}
\includegraphics[width=\imagewidth]{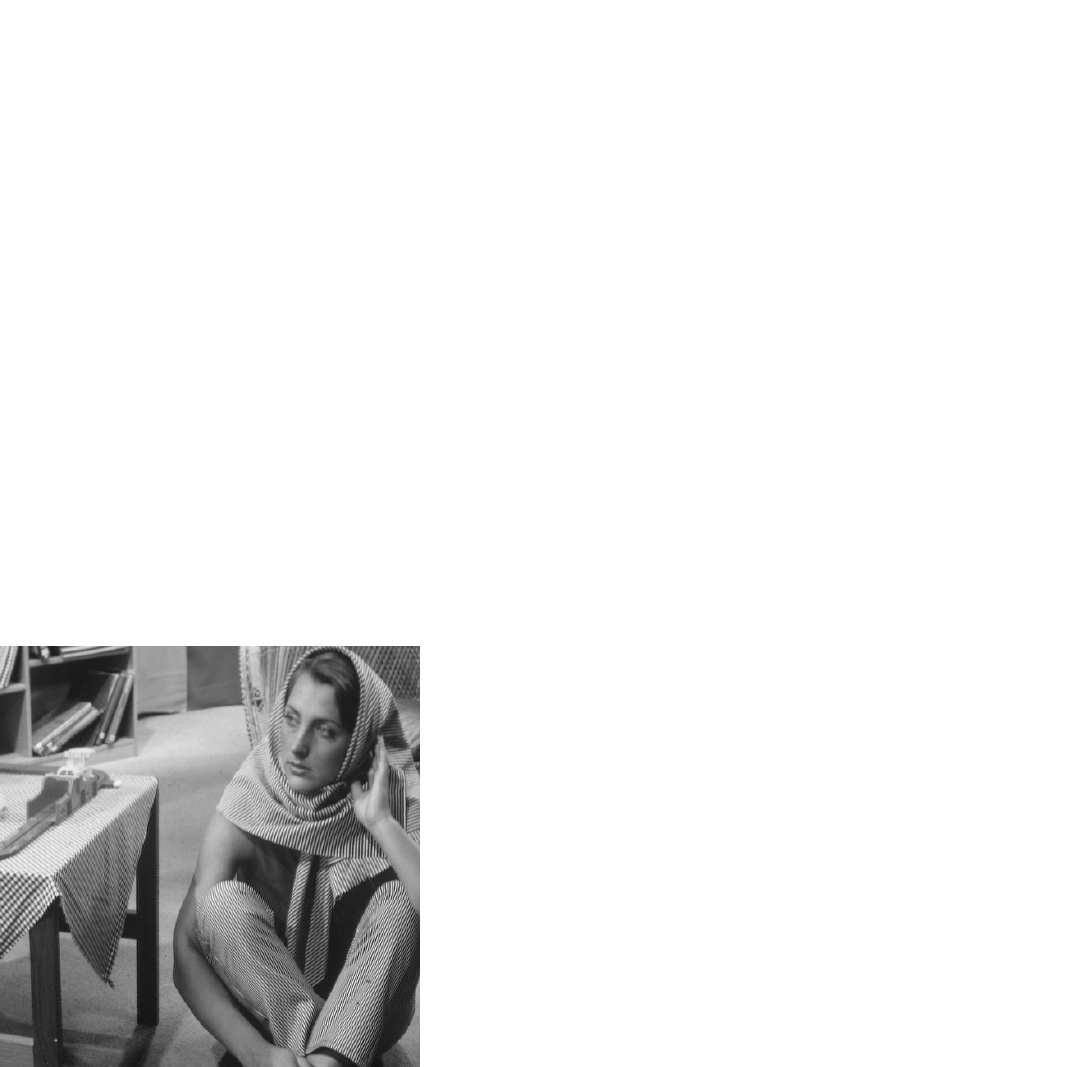}
a)
\end{minipage}
\begin{minipage}[t]{\imagewidth}
\includegraphics[width=\imagewidth]{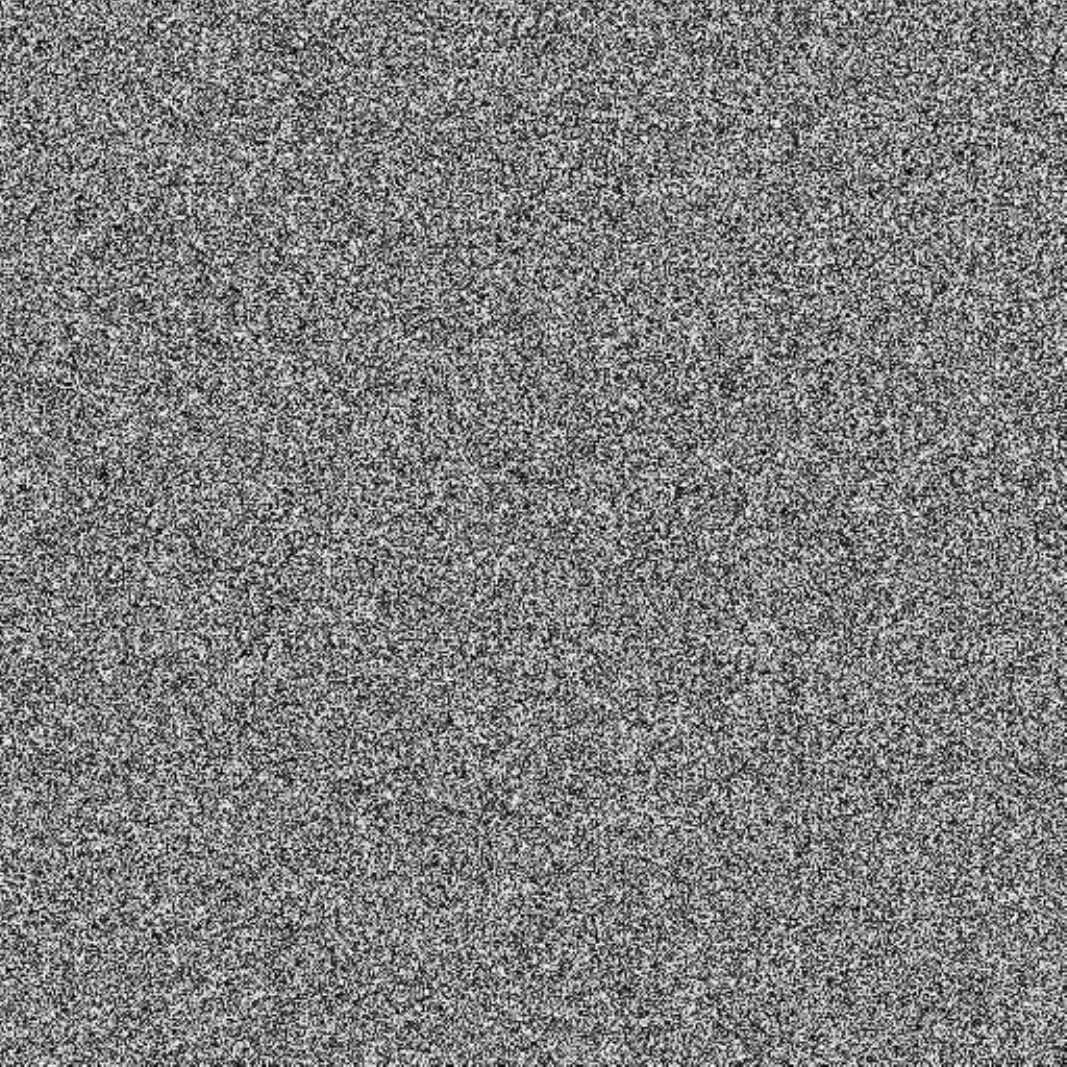}
b)
\end{minipage}
\begin{minipage}[t]{\imagewidth}
\includegraphics[width=\imagewidth]{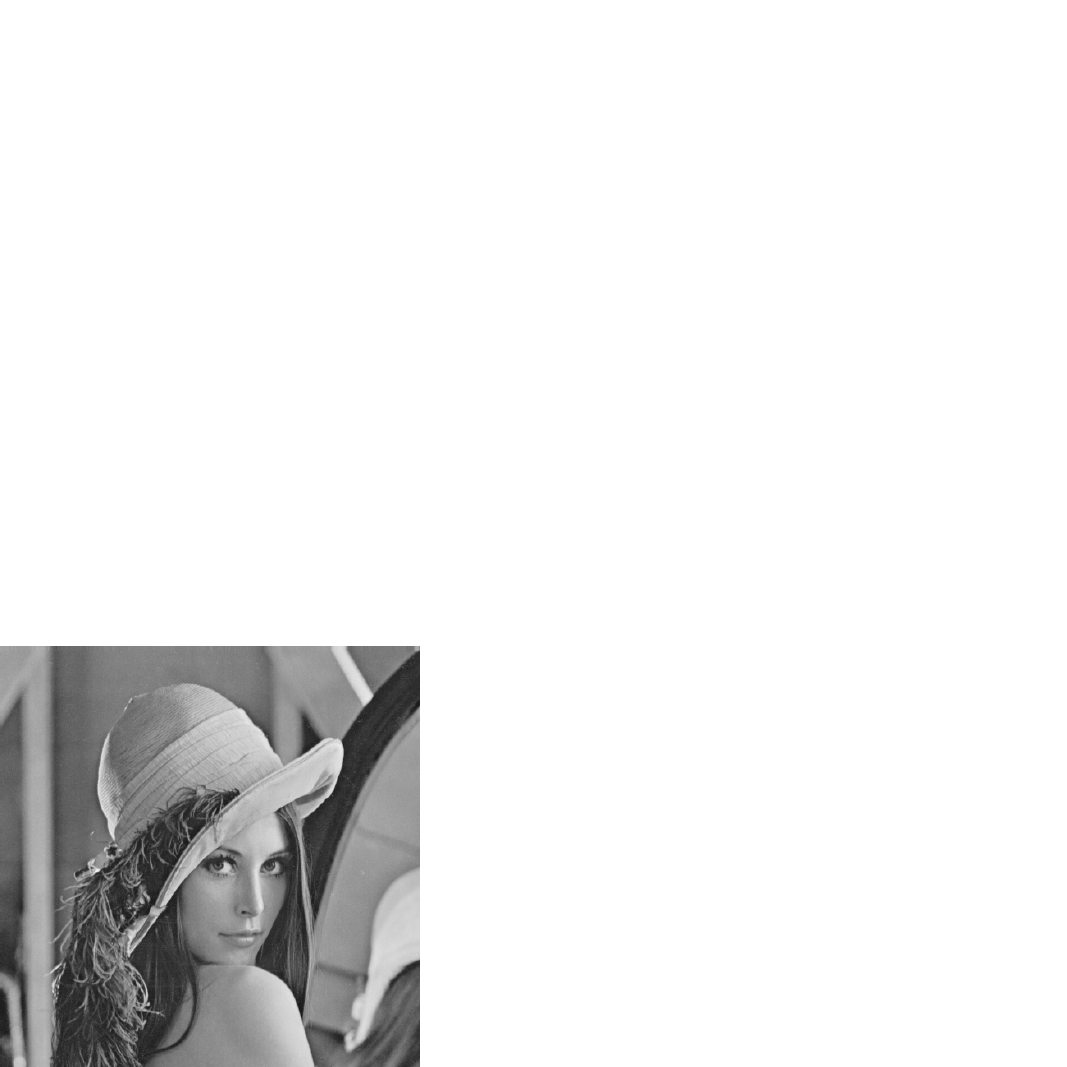}
c)
\end{minipage}
\caption{The known-plaintext attack II: a) the second known plain-image ``Babarra"; b) the cipher-image of plain-image ``Babarra";
c) decryption result of Fig.~\ref{figure:plaintextattack1}c).}
\label{figure:plaintextattack}
\end{figure}

\subsection{Two other security defects}
\label{sec:defect}

In this subsection, two other security defects of the image encryption scheme under study
are discussed.

\begin{itemize}
\item \textit{Low sensitivity with respect to changes of secret key}

In \cite[Sec.~3.3.1]{Zhu:hyperchaotic:OC12}, it was concluded that the image encryption scheme under study is sensitive to changes of secret key from experimental results on some selected secret keys. However, this conclusion lacks a firm ground. Assume there is another secret key generating PRNS $\{k'(i)\}_{i=1}^L$,
where $\{k(i)\oplus k'(i)\}_{i=1}^L=\{0, 128\}$. Let $\{t'(i)\}_{i=1}^L$ and $\{c'(i)\}_{i=1}^L$ denote the corresponding
intermediate sequence and cipher-image, respectively. From Eq.~(\ref{eq:encryptR1}) and Proposition~\ref{prop:xor128}, one has
$t'(1)=t(1)$,
\begin{equation}
t'(i)=
\begin{cases}
t(i)           & \mbox{if } S(i)=0;\\
t(i)\oplus 128 & \mbox{otherwise},
\end{cases}
\end{equation}
for $i=2\sim L$,
where
\begin{equation*}
S(i)=
\begin{cases}
0          & \mbox{if } (k(1)\oplus k'(1)\oplus k(i)\oplus k'(i))=0;\\
1          & \mbox{otherwise},
\end{cases}
\end{equation*}

Then, one further has
\begin{equation*}
c'(i)=
\begin{cases}
c(i)           & \mbox{if } S(L)=1;\\
c(i)\oplus 128 & \mbox{otherwise},
\end{cases}
\end{equation*} for $i=1, 2$.
Based on mathematical deduction, one can obtain
\begin{equation*}
c'(i)=
\begin{cases}
c(i)           & \mbox{if } (S(L)+\sum_{j=2}^{i-1}S(j))\mbox{ is even}, \\
c(i)\oplus 128 & \mbox{otherwise},
\end{cases}
\end{equation*}
for $i=3\sim L$. The above analysis shows that $K'=\{k'(i)\}_{i=1}^L$ is equivalent to $K=\{k(i)\}_{i=1}^L$ with respect to the encryption/decryption procedure
of the least 7 significant bit plane of the plain-image, where $k'(i)\in\{k(i), k(i)\oplus 128\}$. Therefore, there are at least $2^L$ equivalent secret keys for each secret key of the image encryption scheme under study. This serious defect also exists in some other chaotic encryption schemes \cite{Li:chenclass:ISCAS08,AlvarezLi:Rules:IJBC2006,ShujunLi:ChaosBook2011}.

\item \textit{Low sensitivity with respect to change of plain-image}

As well-known in the field of cryptology, sensitivity of encryption results with respect to changes of plaintext is an important property
measuring a secure encryption scheme. This property is especially important for secure image encryption schemes for the following reasons: (1) strong redundancy exists among neighboring pixels of an uncompressed plain-image; (2) a plain-image and its watermarked versions, which generally modify the original image slightly, are often encrypted at the same time. In \cite[Sec.~3.3.2]{Zhu:hyperchaotic:OC12}, it is claimed that the proposed encryption scheme is very sensitive with respect to changes of plain-image. However, the claim is questionable for the following reasons: (1) there is no nonlinear operation, such as S-box, is involved in the whole encryption scheme; 2) there is no operation generating carry bit toward lower level in the whole scheme, so one bit of plain-image can only influence the bits in higher bit planes of the corresponding cipher-image. This defect is very common
for chaos-based encryption schemes \cite[Sec. 2.2]{ShujunLi:ChaosBook2011}.
\end{itemize}

\begin{proposition}
Assume $\alpha$ and $\beta$ are $n$-bit non-negative integers, then
\begin{equation*}
(\alpha\oplus 2^{n-1})\dotplus \beta=(\alpha\dotplus \beta)\oplus 2^{n-1}.
\end{equation*}
\label{prop:xor128}
\end{proposition}
\begin{proof}
First, $\alpha \oplus 2^{n-1}=\alpha \dotplus 2^{n-1}$ can be proven for the following two cases: (1) when
$\alpha\ge2^{n-1}$, one has $\alpha \oplus 2^{n-1}=\alpha-2^{n-1}=\alpha \dotplus 2^{n-1}$; (2) when
$\alpha<2^{n-1}$, one has $\alpha \oplus 2^{n-1}=\alpha+2^{n-1}=\alpha \dotplus 2^{n-1}$. Therefore,
$(\alpha\oplus 2^{n-1})\dotplus \beta=(\alpha\dotplus \beta)\dotplus 2^{n-1}=(\alpha\dotplus \beta)\oplus 2^{n-1}$.
\end{proof}

\section{Conclusion}

This paper re-evaluated the security of a novel image encryption scheme in detail. It was found that the encryption
scheme can be effectively broken with only two known plain-images. Both mathematical proofs
and experimental results were presented to support the proposed attack. In addition, some other security defects
of the encryption scheme were also shown. This paper sets up a good example framework for security
analysis of chaotic cryptosystems.

\section*{Acknowledgement}

This research was supported by the National Natural Science Foundation of China (Nos.~61100216 and 61211130121).

\bibliographystyle{spmpsci}
\bibliography{OC}

\end{document}